%% file: maintaining_human_oversight.tex
\documentclass[11pt]{article}
\usepackage[T1]{fontenc}
\usepackage[utf8]{inputenc}
\usepackage[margin=1in]{geometry}
\usepackage{amsmath,amssymb,amsthm,booktabs,array,graphicx}
\usepackage[round]{natbib}
\usepackage{microtype}
\usepackage{setspace}
\usepackage[hidelinks]{hyperref}
\usepackage{xurl}
\input{numbers_oversight}
\input{numbers_conference}

\newtheorem{proposition}{Proposition}[section]
\newtheorem{lemma}[proposition]{Lemma}
\hypersetup{pdftitle={Maintaining Human Verification Capacity under Automation},pdfauthor={Li Gan and Eric Gan}}
\title{Maintaining Human Verification Capacity\\ under Automation}
\author{Li Gan\thanks{Department of Economics, Texas A\&M University. Email: ganli@tamu.edu.} \ and Eric Gan\thanks{Redwood Research. Email: ericgan2014@gmail.com.}}
\date{September 2026}
\begin{document}
\maketitle
\onehalfspacing

\begin{abstract}
\noindent Human verification depends on expertise that must be maintained before it is needed. This paper links reliance on automated checks, investment in human checking ability, and performance during an interruption. Better checking lowers the error reduction gained from an extra unit of human skill while the checker works. It can therefore reduce the incentive to preserve independent expertise, even when it lowers the best achievable expected cost of maintenance and errors. In an illustration, a more informative checker raises detection while it works from \hoPowWorkA{} to \hoPowWorkC{} percent, but the organization then keeps no routine practice and detects \hoPowOnsetC{} percent of errors when the checker first fails, against \hoPowOnsetA{} percent with a less informative checker. A detection requirement therefore concerns both current capability and its survival until new training becomes effective. Evidence from colonoscopy and aviation documents weaker unaided performance under routine automation, without isolating the mechanism. The framework connects a detection target to an explicit reserve of expertise and a training pipeline. Standard detection tests estimate each quantity and reveal automation bias and silent checker failures. The framework distinguishes the requirement from minimizing expected loss and proposes a longitudinal test.

\end{abstract}

\section{The question that a plan for human checking should answer}

Consider an organization that uses AI to help maintain an important software service. Its staff review changes, but much of the detailed checking comes from an automated analysis service. The checker works well. Reviewing becomes faster, fewer changes require unaided diagnosis, and keeping every reviewer equally practiced seems less urgent. Some training may continue, but the amount and form of practice adapt to the new workflow.

Now the analysis service becomes unavailable, or a change in the software places part of the work outside its validated coverage. The underlying service still needs maintenance. A policy says that experienced staff will review the changes instead. What establishes that those staff can still recognize the relevant errors, and that their ability will last until additional training or help becomes effective?

This is a question about the condition of a safeguard. The presence of reviewers in a process diagram does not answer it. Nor does strong performance when the checker is available. The checking system combines people, tools, information, and organizational choices; deployment of automation can change the future condition of those components.

An argument that people can check the work when the automated checker is unavailable must account for the \emph{maintenance of their verification capacity}: the task-specific capability that lets them detect a specified error without that checker. The proposed framework follows three steps. Specify an error that the organization must detect. Determine what independent evidence remains if a usual checking channel is lost. Then assess whether the available expertise and already committed training can sustain detection until a response becomes effective. The calculation examines one explicit dependency in a checking arrangement.

The framework yields three implications. First, a better routine checker can leave a weaker reserve, because it lowers the error reduction gained from an extra unit of human skill while the checker works (Section~\ref{sec:implication_reserve}). In a numerical illustration, the organization with the most informative checker detects \hoPowWorkC{} percent of errors while the checker works, against \hoPowWorkA{} percent with the least informative checker. It keeps no routine practice, however, and its detection at the onset of an interruption is \hoPowOnsetC{} percent, against \hoPowOnsetA{} percent; because rebuilding starts at once, the gap closes to \hoGapTen{} points within ten periods. Second, average availability does not describe readiness (Section~\ref{sec:implication_availability}). At the same \hoPidPct{} percent disruption frequency, persistent interruptions make a first unit of capability more valuable during a disruption than while the checker works. They also lower the share of harm an organization must bear before it invests at all, from \hoChiCritIid{} to \hoChiCritPers. Third, a promise to retrain is not a reserve (Section~\ref{sec:implication_buffer}): if new formation takes \hoResLwordEn{} periods to count and \hoResLossPct{} percent of capability is lost each period, a requirement of 1 needs a starting stock of \hoResReq.

The same illustration shows that a detection floor is a different objective from minimizing expected cost (Section~\ref{sec:standard}). Holding enough capability to detect at least \hoFloorTargetPct{} percent of errors during any disruption costs the full-harm organization nothing extra with the two less informative checkers and about \hoFloorSharePctC{} percent of its optimized cost with the most informative one. Appendix~\ref{app:numerical} shows this in a two-period version of the decision: better checking makes the same floor weakly more expensive to meet. During an interruption, accepting more false alarms is a second lever: raising the false-alarm rate from \hoFAExPct{} to \hoOpFAHigh{} percent lifts the most exposed organization's detection at onset from \hoPowOnsetC{} to \hoOpHitCHigh{} percent without added capability. Section~\ref{sec:measure} shows how an organization can estimate these quantities from detection tests and training records, and how the same tests reveal automation bias and silent checker failures. Section~\ref{sec:ai_scope} discusses possible applications to AI safety, including errors timed to interruptions.

The software example rests on documented practice. Secure-development guidance distinguishes automated analysis from expert review \citep{nistssdf}, and code can compile even when scanners fail to analyze it \citep{githubscan}, so production can continue while a check is unavailable. How costly interruptions are, how fast expertise decays, and how quickly help arrives are what the framework asks an organization to measure (Section~\ref{sec:measure}).

\section{An established concern and what this paper adds}

The concern is older than generative AI. \citet{bainbridge1983} described the difficulties created when automation removes the ordinary work that sustains the operator's ability to handle abnormal conditions. Loss of practice, the need for training, and the difficulty of relying on people only in emergencies are not new observations. Research on skill retention finds that skills decay with nonuse, faster for cognitive and accuracy-based tasks than for physical ones \citep{arthur1998}. Airline pilots who rely on cockpit automation retain manual control better than the cognitive skills that manual flight requires \citep{casner2014}. In safety-critical professions, retention depends on initial training, practice or refreshers, and task complexity \citep{vlasblom2020}. Regulators in aviation and nuclear power encourage or require recent practice (Section~\ref{sec:existing_evidence}). Work on the out-of-the-loop problem also distinguishes immediate situation awareness from practiced skill \citep{endsley1995}. Automation bias is a further problem: people using an imperfect decision aid miss problems it fails to flag and follow its incorrect advice, experts included, and training or instructions have not prevented it \citep{parasuraman2010}. Users of AI recommendations likewise often accept wrong suggestions \citep{bucinca2021}. An expert who has not followed the current situation, a reviewer who defers to the aid, and a reviewer whose competence has deteriorated face different problems. The framework models the last, and the detection tests in Section~\ref{sec:measure} also measure the second.

Recent AI research makes the connection particularly direct. \citet{mitchell2026} argue that agent deployment can undermine the cognitive abilities needed for effective oversight and discuss design and organizational support for maintaining them. \citet{shen2026} provide experimental evidence that AI assistance can affect learning a programming library, with outcomes differing across interaction patterns. This is relevant evidence about skill formation, not an estimate of long-term deterioration in a particular checking task. Interviews with developers also show that oversight involves planning, monitoring, and review, and that checking agent-generated code can be difficult \citep{dhanorkar2026}. Clinical evidence comes closest to the question here: after AI-assisted detection was introduced in colonoscopy, detection in procedures performed without AI declined \citep{budzyn2025}; Section~\ref{sec:existing_evidence} discusses what that design can establish. What these accounts leave open is how much verification capacity an organization needs when an automated check fails, how long it lasts, and how long it takes to rebuild.

Economic models have begun to treat fallback skill as an investment. \citet{singh2026} study a firm that keeps workers engaged, at a cost in current output, because engagement builds the skill needed when AI fails, and they distinguish AI capability from reliability. \citet{bauer2026} shows how preserved fallback skill, together with liability, can signal competence in markets for expert services. In both, the preserved skill matters when the AI fails. Here human expertise also adds information while the automated checker works, which is why a better working checker can substitute for it (Section~\ref{sec:implication_reserve}). At the level of individual roles, \citet{delachica2026} propose a decision protocol that prices tacit-knowledge erosion and reduced resilience. Standard automation cost--benefit analysis omits these costs; the protocol can recommend augmenting or preserving roles that such analysis would automate. The framework here models how one of those costs, the erosion of verification capacity, arises and how it can be measured.

Research on joint human and automated checking includes the signal-detection analysis of \citet{sorkinwoods1985}, who allowed a human and an automated monitor to observe partially correlated channels. Experiments on assistance for difficult evaluation tasks provide another point of contact \citep{bowman2022}. Both take the human's capability as given; the question here is how it is maintained when automated checking removes the practice that formed it.

Regulation raises the same question. The European Union AI Act requires high-risk systems to be designed so that natural persons can oversee them effectively, including awareness of automation bias (Article~14). It also requires deployers to assign oversight to people with the necessary competence, training, and authority (Article~26) \citep{euai2024}. It does not specify how that competence is maintained when AI performs the routine work, or how it is verified for a period in which an automated component is unavailable. Frameworks for effective human oversight define its architecture, roles, and processes \citep{gaube2026}; this paper concerns one input such arrangements assume, the overseers' capability over time.

The contribution has two parts. The first is a compact analytical connection between a detection target, a maintained capability, and the time needed to replenish that capability. It makes three familiar concerns jointly explicit: improvement in an ordinary checking channel can alter the amount of independent expertise retained; good expected performance need not imply adequate performance during interruption; and training begun in response to a problem cannot protect decisions made before that training becomes effective. The appendices prove the formal results with elementary arguments. The second is a measurement protocol. Because the detection model is standard signal detection, the required capability, its survival rate, and the checker's own contribution can be estimated from detection tests, and the lag from training records; the same tests reveal automation bias and silent checker failures. Both parts rest on one model of detection, so competence, reliability, and emergency readiness become parts of one requirement rather than unrelated ones.

The framework addresses one component of organizational resilience. Broader resilience theory concerns adaptation when systems approach or cross the boundaries of their capabilities \citep{woods2018}. The framework below handles a specified error and a recognized loss of a specified checking channel. Its precision comes from that restriction.

\section{The framework in ordinary terms}
\label{sec:framework}

A checking arrangement has three relevant components. First, there is evidence available without the modeled specialist expertise: elementary tests, observable outcomes, or other checks. Second, there is a checking channel that can supply additional evidence but may become unavailable. Third, there is human expertise that helps distinguish acceptable work from a particular error: the organization's verification capacity. Expertise is useful because it improves detection, not simply because a human has approved the result.

For a given error, more informative evidence improves detection. Human expertise can substitute for some of the evidence supplied by the external checker. When the checker becomes better, less human expertise may be needed to obtain a particular performance level while it is working. This need not hold for every combination of people and tools: some tools make expert interpretation more valuable. In the detection model used here it holds at every level of precision when the false-alarm rate exceeds \hoAlphaCrit, as it does in the illustration. With stricter limits, expertise and checking are complements over a range where detection is below one half and additional evidence has increasing returns (Lemma~\ref{lem:substitution} and Appendix~\ref{app:submodularity}). Even then, a better working checker lowers the investment in expertise whenever the checker, with only the baseline evidence, detects at least half of the errors (Proposition~\ref{prop:checker}).

Expertise also changes over time. Without sufficient formation, part of it is lost. Relevant practice and deliberate training can replenish it, but their benefits arrive later. The basic accounting is:
\begin{quote}
\emph{Future verification capacity equals surviving current capacity plus newly effective practice and training.}
\end{quote}
The formulation allows replacement training. Preserving verification capacity does not require preserving every task through which it was learned: simulations, independent diagnosis, or designed exercises may replace routine production. Whether they do so is a question about learning outcomes. In medical education, simulation with deliberate practice has outperformed traditional clinical training for specific skills \citep{mcgaghie2011}; whether designed exercises also reproduce the range of cases met in routine work has to be assessed task by task. Counting course hours or requiring human clicks does not establish effective formation.

The organization decides how much practice and training to support. Preserving practice is costly because manual execution forfeits the immediate speed, volume, and cost advantages of full AI automation. Decision makers in the model anticipate skill loss, so low reserves need not reflect complacency. When a tool reliably supplies information that people otherwise provided, maintaining human capability becomes less attractive. Whether that reduction is acceptable depends on the safety objective and the alternatives available during an interruption.

The distinction between current assistance and future formation is essential. A person may perform better with a tool today while learning less from the workflow. Conversely, some forms of assistance may preserve or improve learning. A single performance measurement taken while assistance is present cannot distinguish these possibilities.

\section{What existing evidence tells us}
\label{sec:existing_evidence}

Two professions document the outcome this framework concerns. Rules in aviation and nuclear power already require recent practice. Conference records show that evaluation capacity can expand unevenly by role. None of this evidence isolates the dynamic mechanism.

Medicine provides the closest observation. In a multicenter observational study of colonoscopy, the adenoma detection rate in procedures performed without AI fell from 28.4 to 22.4 percent after AI-assisted detection was introduced \citep{budzyn2025}. The study measures the outcome this framework concerns: detection without the automated aid, after routine exposure to it. Being observational, it does not by itself separate a loss of skill from changes in effort, case mix, or other conditions.

Aviation shows the institutional response. Citing an increase in manual handling errors, the Federal Aviation Administration in 2013 encouraged airlines to build manual flight into line operations and training, because continuous use of automation does not reinforce manual flying skills \citep{faa2013safo}. Later that year, a joint government--industry working group recommended that pilots be given opportunities to practice manual flight \citep{parccast2013}. The concern is consistent with simulator evidence that pilots who rely on automation retain their control skills better than the cognitive skills that manual flight requires \citep{casner2014}.

Observational studies also link recent practice to manual flying skill. In simulator tests of airline pilots, manual handling performance was related to recent flying experience \citep{ebbatson2010}. Among 126 randomly selected airline pilots, recent flight practice predicted performance on a manual precision approach more strongly than time since flight school or total flight experience \citep{haslbeck2016}. In 2016 the Department of Transportation's Inspector General found that the agency had no process to ensure that pilots trained to use and monitor automation also maintain proficiency in manual flight, and that it was not well positioned to determine how often pilots fly manually \citep{dotoig2016}. The sequence follows the framework: a regulator responded to skill loss under automation by encouraging retained practice, and its auditor then found that the amount of practice was not being measured.

Some regulators go further and set practice floors. A licensed reactor operator keeps active status only by performing the operator's functions on at least seven 8-hour or five 12-hour shifts per calendar quarter, and one who falls short must complete 40 hours of supervised shift functions before resuming. Every licensee also passes an annual operating test. Requalification includes annual practice of specified abnormal events, on a simulator where needed \citep[10 CFR 55.53 and 55.59]{nrc2026}. An airline pilot may not serve without at least three takeoffs and landings in the type within the preceding 90 days \citep[14 CFR 121.439]{faa2026recency}. Such rules count practice and test proficiency at fixed intervals. The framework says what they should secure: the unaided detection that a stated target requires, over the time before help arrives.

Academic peer review shows that evaluation capacity can differ by role. Published conference records distinguish the reviewers who prepare individual assessments from the area chairs who synthesize those assessments and help reach decisions. Between 2024 and 2025, submissions to ICLR rose by \confIclrSubGrowth{} percent and its reviewer roster more than doubled, in the first year of a reciprocal-reviewing requirement that covered reviewers but not area chairs \citep{iclr2025cfp}. Area-chair numbers grew by \confIclrAcGrowth{} percent, so submissions per listed area chair rose from \confIclrSubPerAcPre{} to \confIclrSubPerAcPost; NeurIPS showed the same pattern (Table~\ref{tab:conference} in Appendix~\ref{app:conference}). The comparison is descriptive; its lesson for this framework is to distinguish the activities on which a decision depends. Preparing a report, reconciling conflicting reports, and making a timely decision can require different capabilities, and a count of people assigned to checking does not establish that every needed capability is present. The capability measure in the framework should therefore be tied to a specified detection task, rather than interpreted as an organization's total reviewer headcount.

\citet{ganlabor2026} supplies a starting point for describing those tasks. The study distinguishes execution---producing a first-pass output---from evaluation---judging whether that output is correct and fit for use---across 19,265 O*NET task statements. The measures are reproducible across model coders and editions of the task data, but they describe work content rather than independently measured expertise. An organization could adapt the distinction to record which production and checking activities people still perform after deployment, what feedback they receive, and what practice a training program replaces. Such records would describe opportunities for learning; they would not measure the resulting capability.

Together, these records give the framework two professions in which unaided performance has been observed to weaken under routine automation and observational evidence that recent practice predicts manual skill. They also provide rules that already require recent practice, an organizational example of capacity that differs by role, and a task-measurement foundation. None isolates the dynamic link: whether a change in practice alters later independent detection, and how long effective training takes to restore that capability. The study proposed below links work records to repeated performance assessments under the relevant checking conditions.

\section{Three implications}

Three implications for checking arrangements follow. The first two are proved in Appendix~\ref{app:numerical} for a two-period version of the organization's decision, using only the curvature of the miss probability in Lemma~\ref{lem:substitution}; the numerical illustration solves the infinite-horizon version. The third is an accounting requirement imposed by training lead times and does not depend on how the organization optimizes.

\subsection{Better routine checks can leave a weaker reserve}
\label{sec:implication_reserve}

The mechanism rests on one sign condition on $m_w$, the probability of missing an error while the checker works:
\begin{equation}
\frac{\partial^2 m_w}{\partial K\,\partial G_w}>0.
\label{eq:substitution}
\end{equation}
A higher precision $G_w$ of the working checker strictly decreases the error reduction gained from an extra unit of human skill $K$. Skill still reduces errors, and more so when the checker fails, but it is worth less while the checker works: at the old level of investment, the last unit no longer covers its cost. Improving the checker therefore weakly lowers the chosen investment in human capability at every starting stock, in both checker states (Proposition~\ref{prop:checker}). Section~\ref{sec:framework} gives the conditions under which \eqref{eq:substitution} holds. The key assumption is that human expertise also adds information while the checker works. If people mattered only when the checker failed, as in models of fallback skill \citep{singh2026}, a better working checker would not by itself lower the value of their expertise.

In the infinite-horizon illustration the lower investment persists over whole histories: on the same sequence of checker states and from the same starting stock, the organization with the better checker holds less capability at every later date (Appendix~\ref{app:numerical}), so it can enter a future interruption with less expertise. Detection while the checker is missing can then be worse, even though access to a better checker reduces the best achievable expected cost. An argument that people can take over when the checker fails therefore needs both ordinary and interruption performance. The improvement changes both the information the tool supplies and how the organization adapts.

A numerical illustration of the infinite-horizon version (Appendix~\ref{app:numerical}) makes the distinction visible (Table~\ref{tab:example}). The organization bears the full harm from undetected errors. Thus the comparison does not rely on managers shifting losses onto someone else. After a prolonged working spell, a stronger checker is associated with less retained practice and less retained human capability. Practice disappears before training does because the first units of training cost less than practice (Appendix~\ref{app:numerical}). The stronger checker raises detection while it works, from \hoPowWorkA{} to \hoPowWorkC{} percent, and lowers it at the onset of an interruption, from \hoPowOnsetA{} to \hoPowOnsetC{} percent: when all three organizations encounter the same loss of checking, their initial detection probabilities differ by up to \hoGapOnset{} percentage points. Measured in misses, the organization with the most informative checker lets \hoMissWorkC{} percent of errors through while the checker works and \hoMissOnsetC{} percent at the onset of an interruption, close to the \hoMissUntrained{} percent that people without task-specific training would miss. The difference is concentrated at the start of the interruption: rebuilding begins at once, and after ten periods the gap is \hoGapTen{} points.

\begin{table}[htbp]
\centering
\caption{An illustrative improvement in checking, the reserve it induces, and recovery}
\label{tab:example}
\small
\begin{tabular}{lccccccc}
\toprule
& & & \multicolumn{5}{c}{Detection (\%)}\\
\cmidrule(lr){4-8}
& Practice & & Checker & \multicolumn{4}{c}{During the interruption}\\
\cmidrule(lr){5-8}
Working checker & kept & Capability & working & Onset & After 1 & After 5 & After 10\\
\midrule
Less informative & \hoPracA & \hoKA & \hoPowWorkA & \hoPowOnsetA & \hoPowOneA & \hoPowFiveA & \hoPowTenA\\
Intermediate & \hoPracB & \hoKB & \hoPowWorkB & \hoPowOnsetB & \hoPowOneB & \hoPowFiveB & \hoPowTenB\\
More informative & \hoPracC & \hoKC & \hoPowWorkC & \hoPowOnsetC & \hoPowOneC & \hoPowFiveC & \hoPowTenC\\
\bottomrule
\end{tabular}
\begin{minipage}{0.96\textwidth}
\small\vspace{5pt}
These are simulated model outcomes, not measured failure rates. Capability is the unaided $d'^2$ that task-specific expertise adds above the baseline of \hoDpBaseSq{} (Section~\ref{sec:measure}). Practice kept is the share of routine work retained for practice during the working spell. Each row starts from its own long-working-spell limit, not a common initial stock or a stationary average; Appendix~\ref{app:illustration} shows that the onset after such a spell is the worst case over histories. Checker working is detection at that limit while the checker works; the interruption then lasts \hoOutageLength{} periods, and the last four columns report detection at its onset and after 1, 5, and 10 periods. Checker precision is respectively \hoGA, \hoGB, and \hoGC, so the checkers' own $d'^2$ are \hoDpCheckSqA, \hoDpCheckSqB, and \hoDpCheckSqC; all other parameters and the disrupted checker are fixed. The organization minimizes maintenance resources plus full expected harm. Better checking still lowers its optimized expected cost when the policies are compared from a common initial stock (Section~\ref{sec:standard}). Appendix~\ref{app:numerical} gives the parameterization and model details.
\end{minipage}
\end{table}

The illustration shows that an organization can make good use of an improved checker and retain a weaker fallback. Evaluating that trade-off requires both the ordinary benefit and the vulnerability during interruptions, together with the probability and cost of interruption.

\subsection{Average availability does not describe readiness}
\label{sec:implication_availability}

Two checking services can be available for the same fraction of time yet have different patterns of interruption. One has many short interruptions; another has longer working spells followed by longer disrupted spells. Training that takes time to become effective has different usefulness in these environments.

The formal model separates the average frequency of disruption from its persistence (Appendix~\ref{app:persistence_proof}). If next period's checker state is independent of today's, today's state changes current detection but not the investment chosen at a given stock: that investment only becomes effective next period. If the state persists, a disruption makes future missing information more likely, and the organization has more reason to rebuild (Proposition~\ref{prop:states}). This result holds because the model separates formation from current exposure. If practice also prevents current errors directly, current checker status can matter for investment even without persistence. In the numerical illustration with the intermediate checker, hold the average disruption frequency at \hoPid{} and raise persistence from zero to \hoGamma. The discounted benefit of a first unit of capability, measured at zero capability with no further investment and full internalization of harm, changes from \hoMBiid{} in both states to \hoMBwork{} while the checker works and \hoMBdisr{} while it is disrupted. It also lowers the internalized share of harm below which the organization never invests, from \hoChiCritIid{} to \hoChiCritPers; the threshold is the cost of a first unit of training, $\kappa$, divided by that benefit in the disrupted state (Appendix~\ref{app:numerical}).

Persistence cuts both ways: longer working spells and longer disrupted spells pull on the decision differently. The useful implication is that a mean-availability statistic is incomplete. Readiness assessment should consider the sequence of conditions in which expertise is retained, lost, and rebuilt. The model also takes the frequency and length of interruptions as given. If maintained capability made them rarer or shorter, for example because practiced staff restore or replace the checker sooner, capability would earn a further return, and the illustration would understate its value.

\subsection{A promise to retrain is not a reserve}
\label{sec:implication_buffer}

Suppose the organization wants to detect a specified error with at least a chosen probability while the usual checker is missing. Given the other evidence that remains, this requirement implies a minimum amount of human capability. If a disruption can occur before additional training becomes effective, the minimum must be met by expertise already available and training already in the pipeline.

The calculation is simple but useful. Assume, only for illustration, that a capability of 1 is required, \hoResSurvivePct{} percent of current capability survives each period, and newly started formation first contributes \hoResLwordEn{} periods after the disruption begins. With no earlier training arriving, the organization must cover the present decision and the next three decisions. Starting at 1 is insufficient: capability falls to \hoResLowOne, \hoResLowTwo, and \hoResLowThree. Starting at about \hoResReq{} covers all four decisions (Figure~\ref{fig:reserve}). Training arriving later may restore capacity, but it cannot change the earlier decisions. The requirement grows geometrically with the lag, as $q^{-(L-1)}$, when no formation can take effect during the interval. That condition matters. If manual checking during the interruption itself maintains skill, or if qualified help arrives sooner, the requirement is smaller; if the interruption ends before the lag has passed, only the periods it lasts need to be covered.

\begin{figure}[htbp]
\centering
\includegraphics[width=0.88\textwidth]{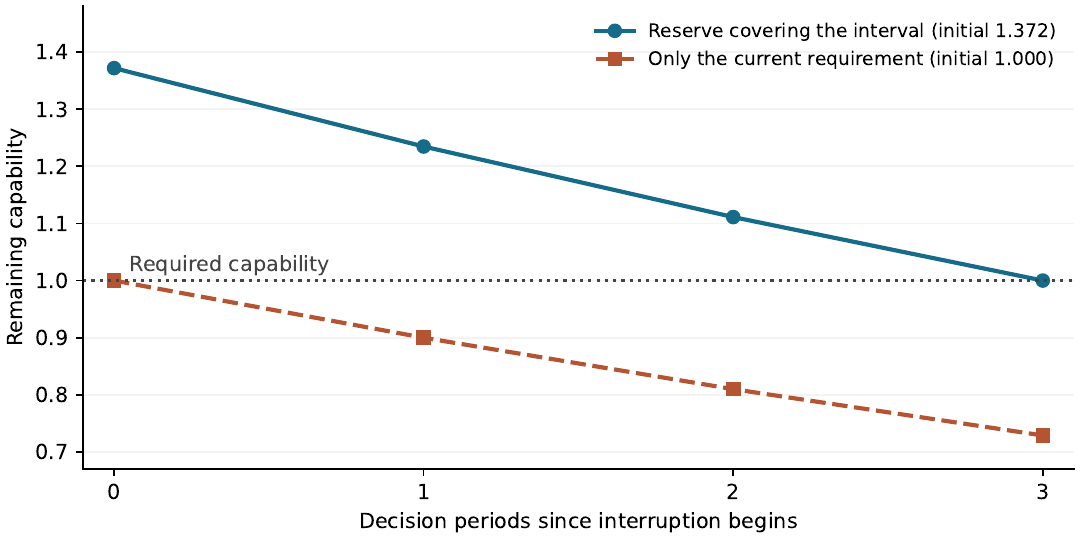}
\caption{Capacity during the interval before new training can help. This arithmetic example assumes \hoResLossPct{} percent loss per period, no committed inflows, and a required capability normalized to 1. The higher starting stock is $1/\hoResQ^3=\hoResReq$. The time unit and loss rate are hypothetical. The figure stops before the first new training becomes effective; it makes no claim about subsequent performance.}
\label{fig:reserve}
\end{figure}

Already scheduled training changes this calculation. If sufficient qualified capacity is due to arrive during the interval, the organization need not hold the entire buffer as existing expertise. Proposition~\ref{prop:reserve} gives the general condition with committed inflows; without them it reduces to $K_t \ge q^{-(L-1)}K_d^*$. It is an accounting requirement for covering the specified interval, not the solution of an optimal investment problem with a longer lag.

\section{Expected performance and a minimum standard are different choices}
\label{sec:standard}

An organization may choose maintenance to minimize its expected harm and the resources used to prevent that harm. A different requirement says that a particular detection probability must be maintained even during an interruption. These objectives need not select the same level of capacity.

The distinction remains when the organization counts every modeled loss. Holding the severity of harm fixed, a less likely interruption receives less weight in an expected-loss calculation. A requirement applying whenever that event occurs still has to be satisfied. Meeting the requirement may therefore require more reserve capacity than the expected-loss optimum. In the illustration, averaged over a long history of working and disrupted periods, the organization with the most informative checker misses \hoAvgMissC{} percent of errors, against \hoAvgMissA{} percent with the least informative one; during interruptions it misses \hoAvgMissDisrC{} percent, against \hoAvgMissDisrA{} percent. These comparisons assume that errors arrive at the same rate whether or not the checker works. If they are more frequent during interruptions, for example because a system change both disables the checker and introduces errors, an expected-loss policy computed with equal rates holds too small a reserve, while a floor, which bounds detection per error, is unaffected. The formal results continue to hold (Appendix~\ref{app:numerical}). The extra capacity is costly; its desirability depends on the purpose and strictness of the requirement. A floor has one advantage of its own: it does not require estimating the probability and persistence of rare interruptions, which an organization's own history seldom reveals, and it still protects when an interruption coincides with unfamiliar work (Section~\ref{sec:measure}).

The companion model gives the cost of such a floor \citep{ganprocuring2026}. Suppose that at every date the organization must hold enough capability to detect at least \hoFloorTargetPct{} percent of errors if the checker is disrupted, and that it starts from a common stock of \hoFloorKzero, which meets the requirement. As working precision rises from \hoGA{} to \hoGB{} to \hoGC, the full-harm organization's optimized cost per period falls from \hoPlanCostA{} to \hoPlanCostB{} to \hoPlanCostC. The floor adds nothing to that cost with the two less informative checkers, whose optimal policies already hold enough capability, and \hoFloorPremPlanC{} per period with the most informative one, about \hoFloorSharePctC{} percent of the optimized cost. The intermediate organization holds just enough: its detection at the onset of an interruption is \hoPowOnsetB{} percent (Table~\ref{tab:example}). Better checking lowers expected cost while making the same floor more expensive to meet; Proposition~\ref{prop:floor} proves this for the two-period version of the decision. These costs start from a common stock, so they are not comparable with the rows of Table~\ref{tab:example}, which start from each organization's own long-working-spell limit.

When an organization bears only some of the damage from errors, it can have an additional reason to underinvest. In the numerical illustration, an organization bearing \hoChiPrivatePct{} percent of the harm, with the intermediate checker, retains capability \hoKPrivB{} rather than \hoKB, and its detection at the onset of an interruption is \hoPowOnsetPrivB{} percent rather than \hoPowOnsetBround{} percent. Such an organization is further from a floor: meeting the \hoFloorTargetPct{} percent requirement would raise its own cost by \hoFloorPremBuyA, \hoFloorPremBuyB, and \hoFloorPremBuyC{} per period with the three checkers. Accountability and liability may address that incentive, but they do not remove the distinction between expected performance and a conditional floor. A fully accountable organization can still fail a floor it was never required to meet.

Human reserves are also only one response. A second checking channel may remain sufficiently informative. The organization may be able to delay the decision safely, suspend the affected service, or obtain independent expertise before the next consequential decision. In these cases the required human reserve can be small or zero. The relevant question is what protection is actually available in time.

A related lever is the operating point. During an interruption an organization can accept more false alarms, escalating or deferring more work, which lowers the capability a floor requires. In the illustration, the intermediate organization detects \hoHitWithout{} percent of errors without its checker at a \hoFAExPct{} percent false-alarm rate and \hoOpHitBHigh{} percent at \hoOpFAHigh{} percent. The most informative organization's detection at the onset of an interruption rises from \hoPowOnsetC{} to \hoOpHitCHigh{} percent. A \hoFloorExPct{} percent floor needs capability \hoKdStarEx{} at the first rate and \hoOpKdHigh{} at the second. Delaying the decision or suspending the service is the extreme case, in which everything is held back. These figures hold capability at its level when the interruption begins; an organization that planned such a mode in advance would also adjust its investment, which the model does not compute. The model provides no general justification for maximizing the human share of work.

\section{What an organization could measure}
\label{sec:measure}

The framework suggests a way to make a claim about verification capacity testable. These are proposed assessment questions, not a validated checklist or a new certification standard. NIST's AI Risk Management Framework already includes proficiency, training, oversight processes, and ongoing evaluation \citep{nistrmf}. This paper adds one structure: a link from those activities to the capability and time interval on which a particular claim depends. Table~\ref{tab:measure} lists what to estimate and how; the paragraphs below explain each step.

\begin{table}[htbp]
\centering
\caption{What to estimate, and how}
\label{tab:measure}
\small
\begin{tabular}{>{\raggedright\arraybackslash}p{0.33\textwidth}>{\raggedright\arraybackslash}p{0.59\textwidth}}
\toprule
Quantity & How to estimate it\\
\midrule
Required capability $K_d^*$ & Choose the detection target $1-\varepsilon$ and false-alarm rate $\alpha$; then $d'^*=z(1-\alpha)+z(1-\varepsilon)$, and $K_d^*$ is $d'^{*2}$ less the baseline $d'^2$ and the disrupted checker's $d'^2$\\
Baseline evidence $b_0^2I_0$ & Unaided $d'^2$ of people without the task-specific training\\
Human capability $K$ & Unaided $d'^2$ less the baseline\\
What the checker adds & Difference in $d'^2$ with and without the checker\\
What the checker could add & The checker's own $d'^2$, scored alone on the same cases\\
Survival rate $q$ & Coefficient on the previous score when the unaided $d'^2$ above baseline is regressed on its previous value and on recorded practice and training, using split-half scores (see text)\\
Lag $L$ & Trainees' time to competence\\
Committed inflows $v$ & Training pipeline records\\
What practice and training add & Their coefficients in the same regression, with practice taken from work records of practice and diagnosis, not job titles\\
\bottomrule
\end{tabular}
\begin{minipage}{0.96\textwidth}
\small\vspace{5pt}
$d'$ is the sensitivity index in \eqref{eq:dprime}, and capability is measured in units of $d'^2$. All tests use new, independently adjudicated cases, which seeded errors can supply, scored for false alarms as well as misses. With \hoNerr{} error cases and \hoNclean{} clean cases, $d'$ has a standard error of about \hoSeDp{} at the illustration's operating point.
\end{minipage}
\end{table}

\paragraph{Define the failure and the remaining evidence.}
The organization should specify the class of errors, when detection must occur, and what information the reviewers would retain during the proposed interruption. Unaided detection should be assessed on the cases an interruption would bring: a checker that fails because work has drifted outside its coverage leaves people with unfamiliar cases. Testing without one checker does not establish independence from every other source of failure. For example, a reviewer and a second model may both rely on the same incorrect reference material.

\paragraph{Measure capability under the relevant conditions.}
Use previously unseen, independently adjudicated cases to assess the specified detection task. Where natural errors are too rare to supply them, seeded errors can: known errors of the specified class are planted in cases that are otherwise verified, as reactor requalification does with simulated abnormal events (Section~\ref{sec:existing_evidence}) and software testing does with planted faults \citep{demillo1978}. Seeded cases measure detection of the seeded classes, so the classes should be those the requirement names. Compare performance with ordinary assistance against performance with the relevant checking channel unavailable. Record false alarms as well as misses; rejecting everything would otherwise look like perfect detection. If the claim is about a delayed response, assess how long a decision takes and whether it arrives in time. The present scalar model measures detection capability rather than response latency, so a time limit belongs in its empirical application or an extended model.

\paragraph{Express the results in the model's units.}
The detection model is the equal-variance Gaussian model of signal detection theory \citep{greenswets1966,macmillan2005}, which \citet{sorkinwoods1985} applied to systems combining automated and human monitors, so its quantities can be estimated from test results. The sensitivity index
\begin{equation}
d'=z(\text{hit rate})-z(\text{false-alarm rate})=b_0\sqrt{I_s(K)},
\label{eq:dprime}
\end{equation}
where $z=\Phi^{-1}$ is the standard normal quantile function, recovers the total precision of the evidence. Because precision adds across independent sources, $d'^2=b_0^2(I_0+NK+G_s)$. Only these products matter, so the unit of capability is a choice; with $b_0^2N=1$, as in the illustration, capability is measured in units of $d'^2$, and $d'^2=b_0^2I_0+K+b_0^2G_s$. Testing the same people with and without the checker therefore separates the checker's contribution, the difference in $d'^2$, from the contribution of human capability and baseline evidence, the unaided $d'^2$. A requirement to detect a share $1-\varepsilon$ of errors at false-alarm rate $\alpha$ without the checker requires $d'^*=z(1-\alpha)+z(1-\varepsilon)$, and hence capability $K_d^*=[d'^{*2}-b_0^2I_0-b_0^2G_d]_+$, which is \eqref{eq:Kdstar} in these units. A higher false-alarm tolerance lowers $d'^*$, which is why accepting more false alarms during an interruption substitutes for capability (Section~\ref{sec:standard}).

In the illustration, the intermediate organization detects \hoHitWithout{} percent of errors without its checker and \hoHitWith{} percent with it, at a \hoFAExPct{} percent false-alarm rate. Its $d'$ is \hoDpWithout{} and \hoDpWith; the checker contributes \hoDpCheckerSq{} to $d'^2$, and human capability with baseline evidence contributes \hoDpHumanSq, of which the baseline accounts for \hoDpBaseSq{} and capability for \hoKB, the value in Table~\ref{tab:example}. Detecting \hoFloorExPct{} percent of errors at the same false-alarm rate would require $d'^*=\hoZfa+\hoZfloor=\hoDpStar$ without the checker, which corresponds to capability $\hoDpStar^2-\hoDpBaseSq=\hoKdStarEx$, about \hoKdRatio{} times what the organization holds.

Test size matters. By the standard approximation to the sampling variance of $d'$ \citep{gourevitch1967}, \hoNerr{} error cases and \hoNclean{} clean cases estimate this organization's $d'$ with a standard error of about \hoSeDp{} and its $d'^2$ with a standard error of about \hoSeDpSq. With \hoNcleanMore{} clean cases the standard error of $d'$ falls to \hoSeDpMore, because at a low false-alarm rate most of the uncertainty comes from the clean cases. Such a test separates the \hoDpHumanSq{} the organization holds from the \hoDpStarSq{} that a \hoFloorExPct{} percent requirement needs, but not a one-period loss of \hoLossOnePeriod; estimating retention needs longer intervals, more cases, or pooling across people. The method comes with two checks. If the checker's contribution to $d'^2$ rises with a person's unaided $d'$, expertise and the checker are complements in precision, and the substitution result may fail. If the slope of the receiver operating characteristic in $z$-coordinates differs from one, the equal-variance assumption fails, and an unequal-variance index should replace $d'$ \citep{macmillan2005}.

\paragraph{Test the checker as well as the people.}
The difference in $d'^2$ above is what the checker adds for these people. Scoring the checker's own output on the same cases gives what it could add, its own $d'^2$. Under the model the two are equal: the $d'^2$ of people using the checker equals their unaided $d'^2$ plus the checker's. The identity applies when the checker reports a graded score; with a pass/fail flag, part of any shortfall reflects the coarse output. Otherwise a shortfall shows that the two sources overlap or that people do not use the checker's evidence fully, for example because they defer to it. Its size measures what overlap and deference cost together, in the units of the model. An excess, like a checker contribution that rises with skill, points to complementarity, in which a combination of human and machine judgments beats either alone \citep{steyvers2022}. Repeating the checker's test on new cases drawn from current work also reveals a silent failure: a fall in its $d'$ shows that it has lost information, for instance because the work has drifted outside its validated coverage, before any outage is declared. From the moment such a failure is recognized, the reserve calculation applies.

\paragraph{Measure retention and effective formation.}
Repeat comparable assessments after different lengths of time in the deployed workflow. Track actual practice, diagnostic work, and training, rather than relying only on job titles or attendance. The task distinction in Section~\ref{sec:existing_evidence} can organize these records, but the learning value of an activity also depends on feedback and the person's involvement. Measure whether training improves later performance on new cases. This separates assisted output today from the ability available tomorrow. In the units of the tests, the accounting of Appendix~\ref{app:numerical} reads $y_{t+1}=q\,y_t+u_t$, where $y$ is the unaided $d'^2$ above the baseline and $u_t$ is what practice and training add. Regressing repeated scores on the previous score and on recorded practice and training therefore estimates the survival rate and the yield of each (Table~\ref{tab:measure}). A retest after an interval with little practice or training is the special case in which formation is zero. Each score carries the sampling error above, which biases the coefficient on the previous score toward zero. Scoring two independent halves of the case set, and using one half's score as an instrument for the other, removes the bias. With tests $k$ periods apart, the coefficient estimates $q^k$. Aviation shows why actual practice must be recorded: after the regulator encouraged manual flight, its auditor found that the agency could not determine how often pilots flew manually \citep{dotoig2016}.

\paragraph{Test the interval before help arrives.}
Document the time needed to restore the checker, develop competence, or bring in qualified support. Trainees' time to competence estimates the lag $L$, erring on the safe side, because new formation can begin to count before a trainee is fully competent. Identify which inflows of expertise are already committed. With $K_d^*$, $q$, and $L$ in hand, Proposition~\ref{prop:reserve} gives the stock the organization must hold when an interruption begins. In the units of the tests and without committed inflows, the unaided $d'^2$ above baseline at the onset of an interruption must be at least $q^{-(L-1)}K_d^*$. An exercise should ask whether the organization can meet its requirement throughout that interval, or whether it must defer some decisions. The purpose is to test the proposed fallback under specified conditions, not to extrapolate from an impressive demonstration with every resource available.

Exercises, rotations, and training are established responses to automation risks. Their value in this framework depends on whether they maintain the measured capability that the detection target requires. A checking arrangement could adopt the three elements of the reactor rules in Section~\ref{sec:existing_evidence}: a minimum volume of unaided checking, a periodic unaided test on adjudicated cases, and a supervised return for reviewers who fall below the minimum. The minimum would be set by the detection target rather than by hours. A successful training program could weaken or eliminate the predicted reduction in reserves after tool adoption. That would be evidence about the mechanism and a useful organizational outcome.

\section{An empirical study that could distinguish the mechanism}

A useful study would follow teams performing a recurring task for which independent adjudication is feasible, or in which errors can be seeded. A checking tool would be rolled out to teams in a randomized order, a stepped-wedge design \citep{hussey2007}: every team eventually receives it, and time since adoption varies across teams. Tasks and evaluation cases would be kept comparable, the order randomized across clusters of teams to limit spillovers, and the analysis would control for calendar time, because the share of teams with the tool rises over time by design. Baseline expertise would be measured before deployment. During deployment, the study would track reliance, practice, training, and routine performance. Controlled assessments would then remove the checking channel to evaluate unassisted detection. One test measures $d'^2$ only to about \hoSeDpSq{} with \hoNerr{} error and \hoNclean{} clean cases (Section~\ref{sec:measure}), so detecting declines of the size in the illustration, \hoLossOnePeriod{} per period, needs many cases per assessment, repeated assessments, or pooling across teams.

The prediction is conditional. Better checking should reduce later independent performance only where it reduces effective skill formation sufficiently, relative to depreciation. A separately randomized training intervention could test whether replacing the lost practice prevents that decline. In the model, holding effective formation fixed isolates the immediate information supplied by a checker from the consequences of adaptation. An empirical study can hold a training protocol fixed, but it must measure whether the resulting learning is comparable; equal training hours do not establish equal formation. Real disruptions alone cannot identify the mechanism, because system changes and difficult tasks may cause both disruptions and poor performance. Before-and-after comparisons of unassisted performance, as in the colonoscopy study, are a cheaper complement. They can document a decline, but not separate lost skill from changes in effort, case mix, or secular trends, which the randomized rollout is designed to do.

Three observations would be especially informative. If independent capability remains stable despite reduced practice, the assumed loss process may be weak in that task. If training maintains capability at modest cost, reliance on external checking need not undermine the reserve. If apparent independent checks share the same errors, the model's additive-information representation is inappropriate.

Studies of AI-assisted learning, the clinical and aviation records, and current practice rules motivate this research. The colonoscopy study documents a decline in unaided detection. The aviation record shows a regulator that encouraged practice without measuring it. Reactor rules already require practice and periodic tests, the conference records show evaluation activities expanding differently by role, and the task measures help describe changes in work and practice. What remains to be measured is the long-run path of task-specific detection capability and the effect of a reserve requirement; the next step is to connect organizational choices and task histories to independently assessed performance, retention, and recovery time.

\section{Possible applications to AI safety}
\label{sec:ai_scope}

One application is the review of research produced by AI systems. Some incorrect claims may appear plausible and require human judgment to detect. If AI systems also check that research, better automated checking lowers the return to human research-checking skill, as in Section~\ref{sec:implication_reserve}. Seeded errors (Section~\ref{sec:measure}) supply adjudicated cases here, for example leakage between training and test data, an inappropriate statistical test, or a sign error, planted in analyses that have been verified; they measure detection of the seeded classes, not of errors no one thought to seed. The clinical and aviation evidence above concerns ordinary automation and does not establish this application.

AI control protocols let a weaker trusted model select which outputs a limited budget of trusted human labor audits, and assume that an audit reveals a successful attack \citep{greenblatt2024}. The framework replaces this perfect audit with one whose accuracy depends on the auditors' maintained capability, so a protocol's safety can decline while its monitor and audit budget stay the same. Control evaluations already generate attacks, which serve as seeded errors: scoring the auditors on the outputs the monitor flags, and repeating the test over time, estimates their capability, its survival rate, and the lag in the units of Section~\ref{sec:measure}. The framework does not model an AI system that chooses when, where, or which errors to introduce, but its conditional requirements still speak to timing. The reserve of Section~\ref{sec:implication_buffer} and the floor of Section~\ref{sec:standard} must hold at every date at which an interruption could begin. Neither uses the probability of interruptions. Concentrating errors of the specified class in interruptions, whether by an attacker or through the common causes of Section~\ref{sec:standard}, therefore does not defeat them. It does defeat a policy chosen to minimize expected loss: in the illustration, the organization with the most informative checker misses \hoAvgMissC{} percent of errors on average but \hoMissOnsetC{} percent at the onset of an interruption after a long working spell, the moment a strategic attacker would choose (Appendix~\ref{app:illustration}). Under that floor it would detect at least \hoFloorTargetPct{} percent at any date. An adversarial guarantee needs more than this. An attacker can choose errors outside the specified class, make many attempts, and exceed the number of outputs people can inspect once the monitor that selects them is unavailable.

Safety cases for frontier AI may rest on control measures that include human monitoring and auditing \citep{clymer2024}. A safety case that lists human intervention as a mitigation should state which objective it adopts and how the required capability $K_d^*$ is maintained, including during the lead time $L$ when an automated monitor is unavailable. Control evaluations count any attack that gets through, so the target applies to each attempt and is strict. The reserve arithmetic shows its cost: in the illustration, detecting \hoStrictPct{} percent of attempts during an interruption requires capability \hoKdStrict{} at a \hoFAExPct{} percent false-alarm rate, \hoKdStrictRatio{} times what the intermediate organization holds. At \hoOpFAHigh{} percent it requires \hoKdStrictHigh. The other responses of Section~\ref{sec:standard}, a second checking channel, a delayed decision, or a suspended service, then carry more of the protection. Control protocols already rely on them, for example by submitting the trusted model's output when suspicion is high.

\section{Conclusion}

The framework concerns a recognized interruption of a checking channel while a specified class of harmful errors remains relevant. The errors are non-adversarial: neither their class nor the availability of checking channels responds to the organization's current ability to detect them, and they arrive at the same rate whether or not the checker works. Section~\ref{sec:standard} notes what changes when they do not. It also assumes that human capability can provide useful independent evidence, that some capability can be maintained through practice or training, and that formation takes time. It does not establish protection against unknown failure modes or errors beyond human evaluation. A failure of the checker enters once it is recognized, and the tests in Section~\ref{sec:measure} can reveal a silent one. The model does not describe how people use a working checker, including automation bias, but the same tests measure its effect. Where detection relies on several distinct skills, each can be treated as its own detection task, with its own requirement, survival rate, and lag. Turnover enters the accounting: departures of trained staff lower the survival rate, and experienced hires are inflows. The framework does not model situation awareness, organizational coordination, or why people leave.

Even within this setting, better technology need not erode useful skill: human and automated checks may be complementary, and training may replenish what routine work no longer supplies. The substitution result requires condition \eqref{eq:substitution}: a higher working precision must decrease the error reduction gained from an extra unit of human skill. This holds at every precision when the false-alarm rate is above \hoAlphaCrit, and at any rate when the working checker, with only the baseline evidence, detects at least half of the errors. The reserve calculation, by contrast, is conditional on a chosen target and an assumed law of capacity loss and replenishment; it does not require proving that AI adoption caused the initial stock to decline.

The contribution is to connect the verification capacity an organization needs during an interruption to the activities that keep it available. Ordinary accuracy, average uptime, reviewer headcount, and promised retraining cannot stand in for that connection. A claim about verification capacity becomes informative when it states what must be detected, what evidence remains when a check is lost, and how the required capacity is maintained until a response takes effect.

\input{references_general.tex}
\appendix
\numberwithin{equation}{section}

\section{Detection, Substitution, and the Reserve Requirement}
\label{app:reserve}

This appendix states the detection model, the condition under which human capability and automated checking are substitutes, and the reserve requirement behind Section~\ref{sec:implication_buffer}. Everything stated here is proved in this paper; an infinite-horizon version of the model is developed in \citet{ganprocuring2026}.

\subsection{Gaussian detection}

A specified error is either absent, $\theta=0$, or present with magnitude $\theta=b_0>0$. The checking arrangement observes a summary $y\sim\mathcal N(\theta,1/I_s(K))$ whose precision combines three independent sources,
\begin{equation}
I_s(K)=I_0+NK+G_s,
\end{equation}
where $I_0>0$ is baseline evidence, $K\ge0$ is task-specific human capability applied over $N\ge1$ inspection steps, and $G_s\ge0$ is the automated checker's precision in state $s\in\{w,d\}$, working or disrupted. With a fixed false-alarm probability $\alpha\in(0,1/2)$, the most powerful test rejects $\theta=0$ when $y\ge a/\sqrt{I_s(K)}$, where $a=\Phi^{-1}(1-\alpha)$. The miss probability is
\begin{equation}
m_s(K)=\Phi\left(a-b_0\sqrt{I_s(K)}\right),
\label{eq:miss}
\end{equation}
and the detection probability is $1-m_s(K)$. The fixed false-alarm rate matters: without a limit or cost on false alarms, rejecting everything would detect every error.

\subsection{Substitution between evidence channels}
\label{app:submodularity}

\begin{lemma}
\label{lem:substitution}
Let $x=b_0\sqrt{I_s(K)}$. Then
\begin{equation}
\frac{\partial^2 m_s}{\partial K\,\partial G_s}=\frac{b_0N\phi(a-x)}{4I_s(K)^{3/2}}\,\bigl(x^2-ax+1\bigr),
\qquad
\frac{\partial^2 m_s}{\partial K^2}=N\,\frac{\partial^2 m_s}{\partial K\,\partial G_s}.
\label{eq:crosspartial}
\end{equation}
If $0<a<2$, both are strictly positive: the miss probability is strictly convex in capability, and a higher checker precision $G_s$ strictly decreases the error reduction gained from an extra unit of capability $K$.
\end{lemma}

\begin{proof}
Write $\zeta=a-x$. Then $\partial m_s/\partial K=-\phi(\zeta)\,b_0N/(2\sqrt{I_s})$. Differentiating with respect to $G_s$, using $\phi'(\zeta)=-\zeta\phi(\zeta)$ and $\partial \zeta/\partial G_s=-b_0/(2\sqrt{I_s})$, gives the first expression, because $1-b_0\zeta\sqrt{I_s}=x^2-ax+1$. Capability enters $I_s$ through $NK$ and the checker enters one for one, so differentiating twice in $K$ multiplies the same expression by $N$. For $0<a<2$, $x^2-ax+1=(x-a/2)^2+1-a^2/4>0$.
\end{proof}

Both derivatives are multiples of one curvature. Capability and the checker enter only through total precision, so
\[
\frac{\partial^2 m_s}{\partial K\,\partial G_s}=N\,\frac{\partial^2 m_s}{\partial I_s^2},\qquad
\frac{\partial^2 m_s}{\partial K^2}=N^2\,\frac{\partial^2 m_s}{\partial I_s^2}.
\]
Substitution between the channels and convexity in capability are therefore the same property, diminishing returns to evidence, and complementarity means increasing returns. Increasing returns arise only where detection is below one half. The roots of $x^2-ax+1$ satisfy $x_\ell x_h=1$ and $x_\ell+x_h=a$, so detection at the upper root is $\Phi(x_h-a)=\Phi(-x_\ell)<1/2$. When detection is low under a strict false-alarm limit, the mean of the test statistic lies below the threshold, and more evidence of either kind moves it toward the threshold, where further evidence adds the most.

The condition $a<2$ is equivalent to $\alpha>1-\Phi(2)\approx\hoAlphaCrit$. It holds at the illustration's $\alpha=\hoAlpha$ but is not universal. At $\alpha=\hoAlphaStrict$, $a=\hoAStrict$ and $x^2-ax+1<0$ for $\hoCompXloStrict<x<\hoCompXhiStrict$; with $b_0=\hoBzero$, this is total precision between \hoCompIloStrict{} and \hoCompIhiStrict, where detection lies between \hoCompPloStrict{} and \hoCompPhiStrict{} percent.

The results in Appendix~\ref{app:numerical} use the property in two places, which then differ. In the disrupted state the checker contributes nothing, so the relevant use is convexity in capability, which Proposition~\ref{prop:states} needs. At the capability levels in the lower two rows of Table~\ref{tab:example}, disrupted-state precision is \hoIdisrB{} and \hoIdisrC, inside the range, so the miss probability is concave in capability there and that convexity fails. In the working state, where the checker's precision varies, the relevant use is substitution, which is all that Propositions~\ref{prop:checker} and~\ref{prop:floor} need. Complementarity requires total precision below \hoCompIhiStrict, which arises only when the checker is weak: never with $G_w$ of \hoGB{} or \hoGC, and with $G_w=\hoGA$ only for capability below \hoCompKhiA. Propositions~\ref{prop:checker} and~\ref{prop:floor} therefore hold at $\alpha=\hoAlphaStrict$ for improvements from $G_w=\hoGB$, a checker that with only the baseline evidence detects \hoDetBaseStrictB{} percent of errors. At $\alpha=\hoAlphaStricter$ the range extends to \hoCompIhiStricter, with detection up to \hoCompPhiStricter{} percent. In general, the working state lies above the range whenever the working checker, with only the baseline evidence, detects at least half of the errors, because detection at the upper end of the range is below one half.

\subsection{The reserve requirement}
\label{app:buffer_proof}

Suppose a requirement says that, while the checker is disrupted, the error must be detected with probability at least $1-\varepsilon$, where $0<\varepsilon\le1/2$. From \eqref{eq:miss}, $m_d(K)\le\varepsilon$ requires total precision $I^*$ and hence capability $K_d^*$, where
\begin{equation}
I^*=\left[\frac{\Phi^{-1}(1-\alpha)+\Phi^{-1}(1-\varepsilon)}{b_0}\right]^2,
\qquad
K_d^*=\frac{[I^*-I_0-G_d]_+}{N},
\label{eq:Kdstar}
\end{equation}
and $[x]_+=\max\{0,x\}$. If $I_0+G_d\ge I^*$, the remaining evidence meets the standard and $K_d^*=0$.

\begin{proposition}[Reserve before new formation matures]
\label{prop:reserve}
Suppose an interruption begins at date $t$ and any newly chosen formation, through retained practice or training, first contributes at date $t+L$, where $L\ge1$. Let $v_0,\ldots,v_{L-2}\ge0$ be inflows committed before the interruption and arriving at dates $t+1,\ldots,t+L-1$, and let capability otherwise survive at rate $q=1-\delta\in(0,1)$. For an interruption that lasts at least until $t+L-1$, capability stays at or above $K_d^*$ at every date $t,\ldots,t+L-1$ if and only if
\begin{equation}
K_t\ \ge\ \max_{0\le j<L}\ q^{-j}\Big[K_d^*-\sum_{\ell=0}^{j-1}q^{j-1-\ell}v_\ell\Big]_+ .
\label{eq:reserve}
\end{equation}
Without committed inflows, the condition is $K_t\ge q^{-(L-1)}K_d^*$.
\end{proposition}

\begin{proof}
Before new formation matures, $K_{t+j}=q^jK_t+\sum_{\ell=0}^{j-1}q^{j-1-\ell}v_\ell$ for $0\le j<L$. Requiring each of these stocks to be at least $K_d^*$ and dividing by $q^j$ gives one condition per date, and \eqref{eq:reserve} collects them. With no committed inflows, the condition for date $t+j$ is $K_t\ge q^{-j}K_d^*$, which is most demanding at $j=L-1$.
\end{proof}

Two assumptions do the work. Manual checking during the interruption is assumed not to add effective capability before $t+L$; if it does, the requirement falls. The interruption is assumed to last until $t+L-1$; a shorter one needs cover only while it lasts. The decision problem in Appendix~\ref{app:numerical} uses $L=1$: formation counts from the next period. The proposition extends the accounting to longer lags without solving the optimal investment policy for them.

\section{The Decision Problem and the Numerical Illustration}
\label{app:numerical}

This appendix proves the paper's formal results in a two-period version of the organization's decision, using only the properties of the detection model in Appendix~\ref{app:reserve}. It then describes the infinite-horizon version that produces the numerical illustration.

\subsection{A two-period decision}

Today the checker is working ($s=w$) or disrupted ($s=d$); a disrupted checker supplies less information, $G_d<G_w$ ($G_d=0$ in the illustration). The organization holds capability $K$. It chooses new formation $u\ge0$, which costs $C(u)$ and counts from tomorrow:
\begin{equation}
K'=qK+u,\qquad q=1-\delta.
\label{eq:formation}
\end{equation}
Tomorrow the checker is disrupted with probability
\begin{equation}
P_{sd}=(1-\gamma)\pi_d+\gamma\,\mathbf 1\{s=d\},
\label{eq:transition}
\end{equation}
and working with probability $P_{sw}=1-P_{sd}$. Here $\pi_d\in(0,1)$ is the long-run share of disrupted periods, and $\gamma\in[0,1)$ measures persistence: with $\gamma=0$ tomorrow's state does not depend on today's, and a larger $\gamma$ makes a disruption today more predictive of one tomorrow. Both $P_{ww}$ and $P_{dw}$ are positive. Tomorrow's expected harm from missed errors is $R\,m_{s'}(K')$, discounted by $\beta$; today's harm does not depend on the choice.

\paragraph{Formation cost.}
Forming capability $u\ge0$ costs $C(u)$, where $C$ is convex and continuously differentiable, with $C(0)=0$ and first-unit cost $C'(0)=\kappa>0$. The simplest example is $C(u)=\kappa u+\tfrac c2 u^2$; the illustration uses the least cost of combining practice and training (Appendix~\ref{app:illustration}). The organization chooses $u\ge0$ to minimize
\begin{equation}
F_s(u)=C(u)+\beta R\sum_{s'}P_{ss'}\,m_{s'}(qK+u).
\label{eq:twoperiod}
\end{equation}
An organization that bears only the share $\chi<1$ of the harm replaces $R$ by $\chi R$. If errors are more frequent while the checker is disrupted, replace $R$ by $R_{s'}$ inside the sums: with $R_d\ge R_w$ the results below still hold, because Propositions~\ref{prop:checker} and~\ref{prop:floor} involve only the working-state term and the slope used for Proposition~\ref{prop:states} becomes $\beta\gamma(R_dh_d-R_wh_w)>0$.

\paragraph{Comparing two decisions.}
The results compare the choices made under two objectives, $F$ and $F^+$, that differ in one respect. Each has a minimizer, because $C$ grows without bound while expected harm is bounded. Let $D(u)=F^+(u)-F(u)$, and suppose that $D$ has a positive slope. Then every minimizer $u^+$ of $F^+$ is at most every minimizer $u$ of $F$: since $F^+(u^+)\le F^+(u)$ and $F(u)\le F(u^+)$, adding the two inequalities gives $D(u^+)\le D(u)$, and $D$ rises strictly, so $u^+\le u$. If moreover $u>0$, then $F'(u)=0<D'(u)=F^{+\prime}(u)$, so $u$ does not minimize $F^+$, and $u^+<u$. The comparison needs no convexity.

\paragraph{The optimal choice.}
Let $h_{s'}(K)=-m_{s'}'(K)$ be the reduction in misses from a marginal unit of capability, and let
\begin{equation}
MB_s(u)=\beta R\sum_{s'}P_{ss'}\,h_{s'}(qK+u)
\label{eq:mb}
\end{equation}
be the marginal benefit of formation. Suppose $0<a<2$. By Lemma~\ref{lem:substitution}, each $m_{s'}$ is strictly convex, so $MB_s$ falls as $u$ rises, while the marginal cost $C'$ does not fall. The objective is therefore strictly convex. Its unique minimizer $u_s$ is zero if $MB_s(0)\le\kappa$; otherwise it is the single point where $C'(u)=MB_s(u)$, which exists because $MB_s$ falls toward zero as capability grows while $C'\ge\kappa>0$. With the quadratic cost, the condition reads $\kappa+cu=MB_s(u)$. Under the same condition, extra capability reduces misses by more when the checker is disrupted: $h_d(K)>h_w(K)$ at every $K$. The miss probability depends on $K$ and $G_s$ only through total precision, and it is convex in total precision (Appendix~\ref{app:submodularity}), so a unit of capability reduces misses by more when total precision is lower. The disrupted state has lower total precision at every $K$.

\subsection{Results}
\label{app:persistence_proof}

\begin{proposition}[Better routine checks]
\label{prop:checker}
Suppose that a higher working precision strictly decreases the error reduction gained from an extra unit of capability, $\partial^2 m_w/\partial K\,\partial G_w>0$ at every capability level, as in \eqref{eq:substitution}. By Lemma~\ref{lem:substitution}, this holds when $0<a<2$, and for any $a$ when the working checker, with only the baseline evidence, detects at least half of the errors. Then raising $G_w$ while holding $G_d$ fixed weakly lowers optimal formation in both states and at every capability level, strictly when formation is positive. Tomorrow's capability is therefore weakly lower, and so is detection if the checker is then disrupted. The organization's optimized expected cost falls.
\end{proposition}

\begin{proof}
The second condition suffices because complementarity requires total precision inside the range of Appendix~\ref{app:submodularity}, where detection is below one half; if detection is at least one half at zero capability, every working-state precision lies above that range. Raise $G_w$ to $G_w^+$, and mark the new objective and choice with a plus. The objective changes by
\[
D(u)=F_s^+(u)-F_s(u)=\beta R\,P_{sw}\big[m_w^+(qK+u)-m_w(qK+u)\big],
\]
whose slope is $\beta R\,P_{sw}\big[h_w(qK+u)-h_w^+(qK+u)\big]$. The positive cross-partial means that the better checker reduces misses by less from each unit of capability, $h_w^+<h_w$, and $P_{sw}>0$ in both states, so $D$ has a positive slope, and the comparison gives the result for formation. Tomorrow's capability $qK+u_s$, and detection without the checker, $1-m_d(qK+u_s)$, fall with it. For the cost, $D(u)<0$ at every $u$, so the minimum of the objective falls.
\end{proof}

\begin{proposition}[Persistence]
\label{prop:states}
Let $0<a<2$ and hold $\pi_d$ fixed.
\begin{enumerate}
\item[(a)] If $\gamma=0$, optimal formation is the same in both states.
\item[(b)] If $\gamma>0$, it is weakly higher when the checker is disrupted today, at every capability level, and strictly higher whenever it is positive in the disrupted state.
\item[(c)] An organization bearing the share $\chi$ of the harm forms no capability, at any capability level or state, if and only if $\chi\le\chi_c$, where
\[
\chi_c=\frac{\kappa}{\beta R\sum_{s'}P_{ds'}\,h_{s'}(0)};
\]
and $\chi_c$ falls as $\gamma$ rises.
\end{enumerate}
\end{proposition}

\begin{proof}
By \eqref{eq:transition}, $P_{dd}-P_{wd}=\gamma$, so the objectives in the two states differ by
\[
F_d(u)-F_w(u)=\beta R\,\gamma\,\big[m_d(K')-m_w(K')\big],\qquad K'=qK+u.
\]
If $\gamma=0$, the two problems coincide, which gives (a). If $\gamma>0$, the slope of $F_w-F_d$ is $\beta R\,\gamma\,[h_d(K')-h_w(K')]>0$, because $h_d>h_w$, and the comparison gives (b). For (c), the marginal benefit of a first unit, $\chi\beta R\sum_{s'}P_{ss'}h_{s'}(qK)$, is largest at $K=0$, because $h_{s'}$ falls with capability, and it is largest in the disrupted state, where it exceeds the working-state value by $\chi\beta R\,\gamma\,[h_d(qK)-h_w(qK)]\ge0$. Because the objective is strictly convex, the organization forms no capability anywhere exactly when this largest value is at most the cost of the first unit, $C'(0)=\kappa$, which is the condition $\chi\le\chi_c$. As $\gamma$ rises, $P_{dd}=(1-\gamma)\pi_d+\gamma$ rises, and since $h_d(0)>h_w(0)$, the denominator of $\chi_c$ rises and $\chi_c$ falls.
\end{proof}

\begin{proposition}[The price of a floor]
\label{prop:floor}
Under the condition of Proposition~\ref{prop:checker}, suppose tomorrow's capability must reach a floor $\bar K$, such as $K_d^*$ in \eqref{eq:Kdstar}, so that detection without the checker meets its target. Let $\Delta$ be the resulting increase in the optimized expected cost. Then $\Delta=0$ when unconstrained formation already reaches the floor, and $\Delta$ weakly increases with $G_w$.
\end{proposition}

\begin{proof}
Let $\bar u=\bar K-qK$, and let $\hat u_s$ minimize $F_s$ subject to $u\ge\bar u$. If an unconstrained minimizer $u_s$ satisfies $u_s\ge\bar u$, the floor does not bind, and $\Delta=0\le\Delta^+$. Otherwise $\Delta=F_s(\hat u_s)-F_s(u_s)$ with $u_s<\bar u$. Raise $G_w$ to $G_w^+$, and mark the new objective and choices with a plus. Because $u_s^+$ minimizes $F_s^+$, and $\hat u_s$ minimizes $F_s$ over a set that contains $\hat u_s^+$, we have $F_s^+(u_s^+)\le F_s^+(u_s)$ and $F_s(\hat u_s)\le F_s(\hat u_s^+)$. Therefore
\[
\Delta^+-\Delta=\big[F_s^+(\hat u_s^+)-F_s^+(u_s^+)\big]-\big[F_s(\hat u_s)-F_s(u_s)\big]\ \ge\ D(\hat u_s^+)-D(u_s),
\]
where $D=F_s^+-F_s$ is the change in expected harm from the better checker, as in the proof of Proposition~\ref{prop:checker}. It rises with $u$, and $\hat u_s^+\ge\bar u>u_s$, so $\Delta^+\ge\Delta$.
\end{proof}

The two-period version describes one decision. It shows how a better checker, persistence, and a floor change that decision; the infinite-horizon version below follows capability over many periods.

\subsection{The infinite-horizon illustration}
\label{app:illustration}

The numbers in the text come from the infinite-horizon version of the same problem, in which the decision recurs every period and formation affects all later periods through the surviving stock:
\begin{equation}
V_s(K)=R\,m_s(K)+\min_{u\ge0}\Big\{C(u)+\beta\sum_{s'}P_{ss'}V_{s'}(qK+u)\Big\}.
\label{eq:bellmanu}
\end{equation}
The two-period problem truncates this recursion after tomorrow: it replaces the continuation value $V_{s'}(K')$ with tomorrow's expected harm $R\,m_{s'}(K')$. The parameters are illustrative, not estimated: $\beta=\hoBeta$, $\delta=\hoDelta$ ($q=\hoQ$), $B=\hoBcost$, $\lambda=\hoLambda$, $\kappa=\hoKappa$, $c=\hoCcost$, $I_0=\hoIzero$, $N=\hoNsteps$, $b_0=\hoBzero$, $\alpha=\hoAlpha$ ($a=\hoA$), and $R=\hoRloss$. The long-run disruption share is $\pi_d=\hoPid$ and persistence is $\gamma=\hoGamma$, so a working period is followed by a disruption with probability $\hoPwd$ and a disrupted one with probability $\hoPdd$. Disrupted precision is $G_d=0$, and the three working precisions are $G_w\in\{\hoGA,\hoGB,\hoGC\}$. The partially internalizing organization in Section~\ref{sec:standard} has $\chi=\hoChiPrivate$. Detection depends on $b_0$, $N$, $I_0$, and $G_s$ only through $b_0^2N$, $b_0^2I_0$, and $b_0^2G_s$, so $b_0$ and $N$ set the unit of capability; here $b_0^2N=1$, so a unit of capability is a unit of $d'^2$ (Section~\ref{sec:measure}).

The illustration's formation comes from two sources, $u=\lambda r+z$: retained practice $r\in[0,1]$, which costs $Br$ because it diverts work from automated execution, and dedicated training $z\ge0$, at cost $\kappa z+\tfrac c2 z^2$. Its $C(u)$ is the least cost of producing $u$. With $\kappa<B/\lambda$ ($\hoKappa<\hoBoverLambda$ here), the first units of training are cheaper than practice: formation up to $z_0=(B/\lambda-\kappa)/c$ uses training alone, formation up to $z_0+\lambda$ adds practice at the constant marginal cost $B/\lambda$, and larger formation adds training again.\footnote{Explicitly, $C(u)=\kappa u+\tfrac c2u^2$ for $u\le z_0$; $C(u)=\kappa z_0+\tfrac c2z_0^2+\tfrac B\lambda(u-z_0)$ for $z_0<u\le z_0+\lambda$; and $C(u)=B+\kappa(u-\lambda)+\tfrac c2(u-\lambda)^2$ beyond.} This $C$ is convex and continuously differentiable with $C'(0)=\kappa$, so every result in Appendix~\ref{app:persistence_proof} applies to it.

This cost explains the practice column of Table~\ref{tab:example}. With $z_0=\hoZzero$, the organization with the least informative checker retains full practice (\hoPracA) and trains \hoTrainA; the intermediate one retains partial practice (\hoPracB), with training at $z_0$; and the most informative one relies on training alone (\hoTrainC), with no routine practice. Total formation divided by $\delta$ gives their working-spell capabilities, \hoKA, \hoKB, and \hoKC. As required formation falls, the organization cuts training above $z_0$ first, then practice, and only then the cheap first units of training, so a better checker removes retained practice before it removes training.

The numerical solution applies value function iteration to \eqref{eq:bellmanu}, starting from zero on \hoGridCoarse{} capability points on $[0,\hoKmax]$ and continuing from that solution on \hoGridFine{} points, with linear interpolation between points. Because $C$ and the interpolated value functions are convex, the choice of $u$ at each point is found by bisection on the derivative of the objective. Iteration stops when successive value functions differ by less than $\hoTol$ at every point. The solution is checked for Bellman residuals, convexity of the value functions, grid refinement, and each action against an independent scalar minimization. It also confirms, at every grid point, the infinite-horizon counterparts of Propositions~\ref{prop:checker} and~\ref{prop:states}: investment, the value functions, and capability paths do not rise as $G_w$ increases from \hoGA{} to \hoGB{} to \hoGC, and investment when the checker is disrupted is at least as high as when it works. In this illustration the ordering across checkers is strict at every grid point.

In the infinite-horizon version, the benefit of a first unit counts the misses it avoids in every later period as it depreciates. Measured with the intermediate checker at zero capability, with no further investment and full internalization, it is \hoMBiid{} in both states without persistence, and \hoMBwork{} while the checker works and \hoMBdisr{} while it is disrupted with persistence. The thresholds $\kappa$ divided by the disrupted-state benefit are $\hoKappa/\hoMBiid\approx\hoChiCritIid$ and $\hoKappa/\hoMBdisr\approx\hoChiCritPers$ (Section~\ref{sec:implication_availability}). The two-period thresholds of Proposition~\ref{prop:states}(c) are higher, because there a first unit is valued for one period only. The numerical solution confirms the persistent-case threshold: at 90 percent of it nothing is invested at any capability level, and at 110 percent investment begins in the disrupted state.

Table~\ref{tab:example} starts each organization at the capability that a long working spell sustains, $K_w^*=u_w(K_w^*)/\delta$, computed from the numerical solution, and then imposes an interruption lasting \hoOutageLength{} periods. The long-run miss rates in Section~\ref{sec:standard} come from simulating the three full-harm policies, starting from zero capability, on the same \hoSimHistories{} histories of checker states, each \hoSimPeriods{} periods long; along every history, capability is strictly lower with a better checker at every later date. Because errors arrive at the same rate in every period, the average of $m_s(K)$ after the first \hoSimBurn{} periods is the share of errors missed. Once the process has settled, capability stays between $K_w^*$ and the level that a long interruption sustains, so the onset column of Table~\ref{tab:example} is the worst case over histories: averaged over the simulated interruptions, the first disrupted period misses \hoAvgMissOnsetC{} percent of errors with the most informative checker, against \hoMissOnsetC{} percent after a long working spell. Figure~\ref{fig:reserve} is a separate arithmetic illustration of Proposition~\ref{prop:reserve} with $q=\hoResQ$, $L=\hoResL$, $K_d^*=1$, and no committed inflows, which gives $K_t\ge\hoResQ^{-3}\approx\hoResReq$. Its loss rate is chosen for transparency and is not a sensitivity analysis of Table~\ref{tab:example}.

\section{Conference Records}
\label{app:conference}

Table~\ref{tab:conference} reports the conference counts summarized in Section~\ref{sec:existing_evidence}.

\begin{table}[h]
\centering
\caption{Different parts of academic evaluation expanded at different rates}
\label{tab:conference}
\small
\begin{tabular}{lrrrr}
\toprule
& \multicolumn{2}{c}{ICLR} & \multicolumn{2}{c}{NeurIPS}\\
\cmidrule(lr){2-3}\cmidrule(lr){4-5}
& 2024 & 2025 & 2024 & 2025\\
\midrule
Submissions & \confIclrSubPre & \confIclrSubPost & \confNeurSubPre & \confNeurSubPost\\
Reviewer roster & \confIclrRevPre & \confIclrRevPost & \confNeurRevPre & \confNeurRevPost\\
Area-chair roster & \confIclrAcPre & \confIclrAcPost & \confNeurAcPre & \confNeurAcPost\\
Submissions per listed area chair & \confIclrSubPerAcPre & \confIclrSubPerAcPost & \confNeurSubPerAcPre & \confNeurSubPerAcPost\\
\bottomrule
\end{tabular}
\begin{minipage}{0.96\textwidth}
\footnotesize\vspace{4pt}
Sources: ICLR 2024 and 2025 fact sheets \citep{iclr2024facts,iclr2025facts}, the NeurIPS 2024 fact sheet \citep{neurips2024facts}, and the NeurIPS 2025 program chairs' report \citep{neurips2025chairs}; main-track counts. ICLR and NeurIPS denote the International Conference on Learning Representations and the Conference on Neural Information Processing Systems. Area-chair counts exclude senior area chairs. The final row divides the displayed counts; it is not an observed mean assignment load. Submission-stage definitions differ across sources, and roster counts do not measure active hours, expertise, or independent information. These are existing descriptive records, not new estimates of AI's effects.
\end{minipage}
\end{table}

\end{document}

%% file: numbers_oversight.tex
\newcommand{\hoBeta}{0.95}
\newcommand{\hoDelta}{0.15}
\newcommand{\hoQ}{0.85}
\newcommand{\hoBcost}{0.5}
\newcommand{\hoLambda}{0.5}
\newcommand{\hoKappa}{0.15}
\newcommand{\hoCcost}{4}
\newcommand{\hoIzero}{1}
\newcommand{\hoNsteps}{4}
\newcommand{\hoBzero}{0.5}
\newcommand{\hoAlpha}{0.05}
\newcommand{\hoA}{1.6449}
\newcommand{\hoRloss}{6}
\newcommand{\hoChiPrivate}{0.3}
\newcommand{\hoChiPrivatePct}{30}
\newcommand{\hoPid}{0.20}
\newcommand{\hoPidPct}{20}

\newcommand{\hoGamma}{0.85}
\newcommand{\hoPwd}{0.03}

\newcommand{\hoPdd}{0.88}
\newcommand{\hoKmax}{12}
\newcommand{\hoGridFine}{4,801}
\newcommand{\hoGridCoarse}{2,401}
\newcommand{\hoGA}{9}
\newcommand{\hoKA}{5.12}
\newcommand{\hoPracA}{1.00}

\newcommand{\hoPowOnsetA}{75.0}
\newcommand{\hoPowOneA}{75.6}
\newcommand{\hoPowFiveA}{77.1}
\newcommand{\hoPowTenA}{77.8}

\newcommand{\hoGB}{25}
\newcommand{\hoKB}{2.46}
\newcommand{\hoPracB}{0.31}

\newcommand{\hoPowOnsetB}{50.0}
\newcommand{\hoPowOneB}{56.7}
\newcommand{\hoPowFiveB}{69.7}
\newcommand{\hoPowTenB}{74.6}

\newcommand{\hoIdisrB}{10.8}
\newcommand{\hoGC}{64}
\newcommand{\hoKC}{0.48}
\newcommand{\hoPracC}{0.00}

\newcommand{\hoPowOnsetC}{21.5}
\newcommand{\hoPowOneC}{36.2}
\newcommand{\hoPowFiveC}{62.6}
\newcommand{\hoPowTenC}{71.5}

\newcommand{\hoIdisrC}{2.9}
\newcommand{\hoOutageLength}{30}
\newcommand{\hoZzero}{0.21}
\newcommand{\hoBoverLambda}{1}
\newcommand{\hoTrainA}{0.27}

\newcommand{\hoTrainC}{0.07}
\newcommand{\hoTol}{2\times10^{-10}}
\newcommand{\hoGapOnset}{53}
\newcommand{\hoGapTen}{6}
\newcommand{\hoKPrivB}{0.57}
\newcommand{\hoPowOnsetPrivB}{23}
\newcommand{\hoPowOnsetBround}{50}
\newcommand{\hoMBiid}{2.46}
\newcommand{\hoMBwork}{1.98}
\newcommand{\hoMBdisr}{4.38}
\newcommand{\hoChiCritIid}{0.061}
\newcommand{\hoChiCritPers}{0.034}
\newcommand{\hoAlphaCrit}{0.0228}
\newcommand{\hoAlphaStrict}{0.01}
\newcommand{\hoAStrict}{2.326}
\newcommand{\hoCompXloStrict}{0.569}
\newcommand{\hoCompXhiStrict}{1.757}
\newcommand{\hoCompIloStrict}{1.30}
\newcommand{\hoCompIhiStrict}{12.35}
\newcommand{\hoCompPloStrict}{3.9}
\newcommand{\hoCompPhiStrict}{28.5}
\newcommand{\hoCompKhiA}{0.59}
\newcommand{\hoDetBaseStrictB}{58.8}
\newcommand{\hoAlphaStricter}{0.001}

\newcommand{\hoCompIhiStricter}{29.66}

\newcommand{\hoCompPhiStricter}{35.7}
\newcommand{\hoResQ}{0.9}
\newcommand{\hoResSurvivePct}{90}
\newcommand{\hoResLossPct}{10}
\newcommand{\hoResL}{4}
\newcommand{\hoResLwordEn}{four}

\newcommand{\hoResReq}{1.372}
\newcommand{\hoResLowOne}{0.90}

\newcommand{\hoResLowTwo}{0.81}

\newcommand{\hoResLowThree}{0.729}

\newcommand{\hoDpWithout}{1.645}
\newcommand{\hoDpWith}{2.993}
\newcommand{\hoHitWithout}{50.0}
\newcommand{\hoHitWith}{91.1}
\newcommand{\hoDpCheckerSq}{6.25}
\newcommand{\hoDpHumanSq}{2.71}

\newcommand{\hoFloorExPct}{90}
\newcommand{\hoFAExPct}{5}
\newcommand{\hoZfa}{1.645}
\newcommand{\hoZfloor}{1.282}
\newcommand{\hoDpStar}{2.926}
\newcommand{\hoKdStarEx}{8.31}
\newcommand{\hoKdRatio}{3.4}
\newcommand{\hoDpBaseSq}{0.25}
\newcommand{\hoPowWorkA}{86.8}
\newcommand{\hoDpCheckSqA}{2.25}

\newcommand{\hoPowWorkB}{91.1}
\newcommand{\hoDpCheckSqB}{6.25}

\newcommand{\hoPowWorkC}{99.3}
\newcommand{\hoDpCheckSqC}{16}

\newcommand{\hoMissWorkC}{0.7}
\newcommand{\hoMissOnsetC}{78.5}
\newcommand{\hoMissUntrained}{87.4}
\newcommand{\hoSimHistories}{4,000}
\newcommand{\hoSimPeriods}{6,000}
\newcommand{\hoSimBurn}{1,000}
\newcommand{\hoAvgMissA}{15.0}
\newcommand{\hoAvgMissDisrA}{23.0}

\newcommand{\hoAvgMissC}{8.8}
\newcommand{\hoAvgMissDisrC}{41.3}
\newcommand{\hoAvgMissOnsetC}{71.8}
\newcommand{\hoNerr}{100}
\newcommand{\hoNclean}{100}
\newcommand{\hoNcleanMore}{400}
\newcommand{\hoSeDp}{0.25}
\newcommand{\hoSeDpSq}{0.8}
\newcommand{\hoSeDpMore}{0.16}
\newcommand{\hoDpStarSq}{8.56}
\newcommand{\hoLossOnePeriod}{0.37}

\newcommand{\hoOpFAHigh}{20}
\newcommand{\hoOpHitBHigh}{78.9}
\newcommand{\hoOpHitCHigh}{50.6}
\newcommand{\hoOpKdHigh}{4.26}

\newcommand{\hoStrictPct}{99}
\newcommand{\hoKdStrict}{15.5}
\newcommand{\hoKdStrictHigh}{9.8}
\newcommand{\hoKdStrictRatio}{6.3}
\newcommand{\hoFloorTargetPct}{50}

\newcommand{\hoFloorKzero}{2.5}
\newcommand{\hoPlanCostA}{1.835}

\newcommand{\hoFloorPremBuyA}{0.030}
\newcommand{\hoPlanCostB}{1.203}

\newcommand{\hoFloorPremBuyB}{0.117}
\newcommand{\hoPlanCostC}{0.701}
\newcommand{\hoFloorPremPlanC}{0.106}
\newcommand{\hoFloorPremBuyC}{0.190}
\newcommand{\hoFloorSharePctC}{15}

%% file: numbers_conference.tex
\newcommand{\confIclrSubPre}{7,262}
\newcommand{\confIclrSubPost}{11,603}
\newcommand{\confIclrSubGrowth}{59.8}
\newcommand{\confIclrRevPre}{8,950}
\newcommand{\confIclrRevPost}{18,325}

\newcommand{\confIclrAcPre}{624}
\newcommand{\confIclrAcPost}{823}
\newcommand{\confIclrAcGrowth}{31.9}

\newcommand{\confIclrSubPerAcPre}{11.6}

\newcommand{\confIclrSubPerAcPost}{14.1}

\newcommand{\confNeurSubPre}{15,671}
\newcommand{\confNeurSubPost}{21,575}

\newcommand{\confNeurRevPre}{13,640}
\newcommand{\confNeurRevPost}{20,518}

\newcommand{\confNeurAcPre}{1,393}
\newcommand{\confNeurAcPost}{1,663}

\newcommand{\confNeurSubPerAcPre}{11.2}

\newcommand{\confNeurSubPerAcPost}{13.0}